\documentclass[a4paper,12pt]{article}
\pdfpagewidth
\paperwidth
\pdfpageheight
\paperheight
\usepackage[english]{babel} 
\usepackage{epsfig}
\usepackage{fancyhdr} 
\usepackage{amsmath,amssymb}
\usepackage{amscd} 
\usepackage[T1]{fontenc} 
\usepackage[utf8]{inputenc} 
\usepackage[usenames,dvipsnames]{xcolor}
\usepackage{graphicx,color,listings}
\usepackage{hologo}
\usepackage{geometry}
\usepackage{rotating}
\usepackage{caption}
\usepackage{amsthm}
\usepackage{stackengine,scalerel}
\usepackage{hyperref}

\def\F{\mathbb F}
\def\Fq{\mathbb{F}_q}

\def\PG{{\rm{PG}}}

\def\cC{\mathcal{C}}
\DeclareMathOperator\rk{\mathrm{rk}}
\DeclareMathOperator\Aut{\mathrm{Aut}}

\DeclareMathOperator\pgammal{\mathrm{P}\Gamma{}\mathrm{L}}
\def\GL{{\rm{GL}}}

\theoremstyle{plain} 
\newtheorem{thm}{Theorem}[section] 
 
\newtheorem{lem}[thm]{Lemma}

\theoremstyle{defn}
\newtheorem{defn}[thm]{Definition}

\newtheorem{corollary}[thm]{Corollary} 
\theoremstyle{boldremark}
\newtheorem{remark}[thm]{Remark}

\title{Non-linear MRD codes\\
from cones over exterior sets}
\date{}

\author{Nicola Durante \,\,\,
Giovanni Giuseppe Grimaldi \,\,\,
Giovanni Longobardi}

\begin{document}
\maketitle

\begin{abstract}
\noindent By using the notion of a $d$-embedding $\Gamma$ of a (canonical) subgeometry $\Sigma$ and of exterior sets with respect to the $h$-secant variety $\Omega_{h}(\mathcal{A})$ of a subset $\mathcal{A}$, $ 0 \leq h \leq n-1$, in the finite projective space $\PG(n-1,q^n)$, $n \geq 3$, in this article we construct a class of non-linear $(n,n,q;d)$-MRD codes for any $ 2 \leq d \leq n-1$. A code of this class $\cC_{\sigma,T}$, where $1\in T \subseteq \F_q^*$ and $\sigma$ is a generator of $\mathrm{Gal}(\F_{q^n}|\F_q)$, arises from a cone of $\PG(n-1,q^n)$ with vertex an $(n-d-2)$-dimensional subspace over a maximum exterior set $\mathcal{E}$ with respect to $\Omega_{d-2}(\Gamma)$. We prove that the codes introduced in \cite{Coss,DonDur,Sic} are suitable  punctured ones of $\cC_{\sigma,T}$ and we solve completely the inequivalence issue for this class showing that $\cC_{\sigma,T}$ is neither equivalent nor adjointly equivalent to the non-linear MRD codes $\cC_{n,k,\sigma,I}$, $I \subseteq \F_q$, obtained in \cite{Otal}.

\end{abstract}
\\

\section{Introduction}

From the 1970s to the present day, rank distance codes have been deeply studied as useful tool to correct errors and erasures in networks, and for the increasing applications which they have in
storage systems \cite{Rawat}, cryptosystems \cite{Crypto}, space-time codes \cite{spacetime}
and random linear network coding \cite{Silva}.
Moreover, they are linked to many structures studied in finite geometry such as semifields, linear sets, splitting dimensional dual hyperovals and $q$-polymatroids; for more details see \cite{Gorla, PolZu, Shee2} and the references
therein.
Delsarte \cite{Delsarte} and Gabidulin \cite{Gab0} introduced, independently, rank distance codes as sets of matrices over a finite field equipped with the distance $d(A,B)=\rk(A-B)$. 
The most studied ones are codes with the largest possible size. These are called \textit{maximum rank distance codes}, \textit{MRD codes} for short.\\
Rank distance codes that are very rich in structure, for instance subgroups or subspaces of matrices, are the most studied in literature. While only few classes of non-additive and non-linear are known for almost any choices of parameters.\\
In this paper, we will give a geometric construction of a new family of these objects. These codes will be obtained from a cone of a projective space whose vertex is a proper subspace and whose base is a special pointset skew with the vertex. The base of such a cone is linked to the non-linear MRD codes constructed by Cossidente {\em et al.} \cite{Coss}, Durante and Siciliano \cite{Sic} and Donati and Durante \cite{DonDur}. These latter can be obtained in turn by puncturing codes in the relevant family. Finally, we will show that this class of codes is not equivalent to that one constructed by Otal and {\"O}zbudak in \cite{Otal}.\\
The article is structured as follows: in Section \ref{sec2}, some basic notions on rank distance codes will be recalled and the state of the art of the known MRD codes will be traced. In Section \ref{sec3} and \ref{Cfset}, the geometric setting will be introduced and the codes exhibited in \cite{DonDur} by Donati and Durante recalled. In Section \ref{sec5}, a geometrical configuration about cones over some particular subsets of the finite projective space will be investigated and this will give rise to our  class of non-linear codes exhibited in Section \ref{sec6}. Finally in Section \ref{sec7}, we will show that this class is effectively new,  it contains a class of linear MRD codes, and, in certain sense, a class of non-linear MRD ones already known.

\section{Preliminary results and notions}\label{sec2}

Let $\Fq$ be the finite field of $q$ elements, $q$ a prime power, and let $\mathbb{F}_q^{m \times n}$ be the set of $m \times n$ matrices with entries over $\Fq$ endowed with \textit{rank distance} 
$$d: (A,B) \in \mathbb{F}_q^{m \times n} \times \mathbb{F}_q^{m \times n} \longrightarrow \rk(A-B) \in \{0,1,\ldots, \min\{m,n\}\}.$$
A subset of $\mathbb{F}_q^{m \times n}$, including at least two elements is called a {\em rank
distance code}. The {\em minimum distance} $d(\mathcal{C})$ of a code $\mathcal{C}$ is naturally defined by 
\begin{equation*}
    d(\mathcal{C}) =
\min\{d(A, B) : A, B \in \mathcal{C}, A \ne B\}.
\end{equation*}
If $d=d(\mathcal{C})$, we will say that $\mathcal{C}$ is an $(m, n, q;d)$-rank distance code. An $(m, n, q;d)$-rank distance code is {\em additive} if it is an additive subgroup of $\mathbb{F}_q^{m \times n}$. An additive code  is $\F_q$-{\em linear} if it is a subspace of $\mathbb{F}_q^{m \times n}$ seen as a vector space over $\mathbb{F}_q$. A code that it is not an $\F_q$-subspace is called a \textit{non-linear} code.

It is well-known that the size of an $(m, n, q; d)$-rank distance code $\mathcal{C}$ satisfies the
\textit{Singleton-like bound}, see e.g. \cite[Theorem 5.4]{Delsarte}:
$$
|\mathcal{C}|\leq q^{\max\{m,n\}(\min\{m,n\}-d + 1)}.
$$
When this bound is achieved, $\mathcal{C}$ is called an $(m, n, q; d)$-\textit{maximum
rank distance} code, or in short $(m, n, q; d)$-\textit{MRD} code.
The \textit{adjoint code} of $\mathcal{C}$ is defined as the code made up by the transpose matrices belonging to $\mathcal{C}$:
\begin{equation*}
    \mathcal{C}^t=\{C^t: C \in \mathcal{C}\}
\end{equation*}
where the superscript $t$ stands for the matrix transposition.
By the classification of rank isometries of $\mathbb{F}^{m \times n}_q$
\cite[Theorem 3.5]{Wan}, two rank distance codes $\mathcal{C}, \mathcal{C}^{\prime} \subseteq \mathbb{F}_q^{m \times n}$, $m,n \geq 2$, are called {\em equivalent} if there exist $P\in \GL(m, q)$, $Q \in \GL(n,q)$, $R \in \mathbb{F}_q^{m \times n}$ and a field automorphism $\rho \in \Aut(\mathbb{F}_q)$ such that
$$\mathcal{C}^{\prime}=P\mathcal{C}^{\rho}Q+R=\left\{ PC^{\rho}Q+R: C \in \mathcal{C} \right\}.$$
When $m=n$, in addition to being equivalent, two codes are said \textit{adjointly equivalent} if
$$\mathcal{C}^{\prime}=P(\mathcal{C}^t)^{\rho}Q+R=\left\{ P(C^t)^{\rho}Q+R: C \in \mathcal{C} \right\}.$$

 If both $\mathcal{C}$ and $\mathcal{C}^{\prime}$ are additive, then  we may assume that $R$ is the zero matrix. 
The code $\mathcal{C}^{[u]} \subseteq \mathbb{F}^{(m-u) \times n}_q$ obtained from $\mathcal{C} \subseteq \F_q^{m \times n}$ by  deleting the last $u$ rows, $ 1 \leq u \leq m-1$, is called a \textit{punctured code} of $\mathcal{C}$ . In  \cite[Corollary 7.3]{coverradius}, it is showed that if $\mathcal{C} \subseteq \F_{q}^{m \times n}$, $m \leq n$, is an MRD code then $\mathcal{C}^{[u]}$ is MRD as well. \\
The {\em left} and {\em right idealisers} of a code $\cC \subseteq \Fq^{m \times n}$ are defined as the sets
\begin{equation*}
I_L(\cC)= \left\{ P \in \Fq^{m \times m} : PC \in \cC, \forall C \in \cC \right\}
\end{equation*}
and
\begin{equation*}I_R(\cC)=\left\{ Q \in \Fq^{n \times n} : CQ \in \cC, \forall C \in \cC \right\},
\end{equation*}
respectively.\\
Although the rank distance codes are subsets of matrices, they can be represented in a particular setting, that of $\sigma$-linearized polynomials.
From now on, suppose $m=n$ and let $\sigma : x \in \F_{q^n} \longrightarrow x^{q^s} \in \F_{q^n}$ be a field automorphism of $\F_{q^n}$ with $\gcd(s,n)=1$.

A $\sigma$-\textit{linearized polynomial with coefficients over} $\mathbb{F}_{q^n}$ is a polynomial of the form $\alpha=\sum_{i=0}^\ell \alpha_i X^{\sigma^i} \in \mathbb{F}_{q^n}[X]$, $\ell \in \mathbb{N}$. If $\alpha_\ell \ne 0$, the integer $\ell$ is called the {\em $\sigma$-degree} of $\alpha$ and it will be denoted $\deg_{\sigma}\alpha$ or $\deg_{q^s}\alpha$. The set $\mathcal{L}_{n,q,\sigma}$ of $\sigma$-linearized polynomials with coefficients over $\F_{q^n}$ equipped with the usual sum, the scalar multiplication by an element of $\F_{q^n}$ and the map composition is an algebra over $\F_{q^n}$. While, the set
\begin{equation}
\tilde{\mathcal{L}}_{n,q,\sigma}= \Biggl \{\sum_{i=0}^{n-1} \alpha_i X^{\sigma^i} : \alpha_0,\ldots,\alpha_{n-1} \in \mathbb{F}_{q^n} \Biggr\}
\end{equation}
endowed with the sum as in $\mathcal{L}_{n,q,\sigma}$, the multiplication by an element of $\F_{q}$ and the composition  modulo $X^{q^{ns}}-X$ is an algebra over $\F_q$ isomorphic to the algebra of the endomorphisms of $\mathbb{F}_{q^n}$ seen as vector space over $\mathbb{F}_q$. Indeed, it is well known that any $\F_q$-linear endomorphism of $\F_{q^n}$ can be represented uniquely as a $\sigma$-polynomial belonging to $\tilde{\mathcal{L}}_{n,q,\sigma}$, \cite[Chapter 3 ]{Lidl}.
Then any rank distance code $\mathcal{C}$, $|\mathcal{C}| \geq 2$, can be seen as a suitable subset of $\tilde{\mathcal{L}}_{n,q,\sigma}$.

So, the definitions of transpose matrix, adjoint code and equivalence between codes can be naturally reformulated in this environment. Let $\alpha=\sum_{i=0}^{n-1} \alpha_i X^{\sigma^i} \in   \tilde{\mathcal{L}}_{n,q,\sigma}$, then the \textit{adjoint polynomial} of $\alpha$ is defined as
\begin{equation*}
    \hat{\alpha}=\sum_{i=0}^{n-1}\alpha^{\sigma^{n-i}}_iX^{\sigma^{n-i}}
\end{equation*} and if $\mathcal{C} \subseteq \tilde{\mathcal{L}}_{n,q,\sigma} $ is a rank distance code, $\mathcal{C}^t= \{ \hat{\alpha} : \alpha \in \mathcal{C}\}$. 
Moreover, two rank codes $\mathcal{C}$ and $\mathcal{C}^\prime$ are equivalent or adjointly equivalent if there
exists $(f,\rho,g,h)$ such that $f,g,h \in \tilde{\mathcal{L}}_{n,q,\sigma}$, with $f$ and $g$ permutation polynomials, and $\rho \in \Aut(\mathbb{F}_q)$ such that
\begin{center}
$\mathcal{C}^\prime = \{f \circ \alpha^\rho \circ g+h \colon  \alpha\in \mathcal{C}\} $ \,\,or\,\, $\mathcal{C}^\prime = \{f \circ \alpha^\rho \circ g+h \colon  \alpha\in \mathcal{C}^t\},$
\end{center}
respectively. Here the automophism $\rho$ acts only over the coefficients of a polynomial $\alpha$ in $\cC$ or $\cC^t$, respectively.
A map 
$$(f,\rho,g,h): \alpha \in \mathcal{C} \longrightarrow f \circ \alpha^\rho \circ g + h \in \mathcal{C}^\prime$$
is called an \textit{equivalence} between $\mathcal{C}$ and $\mathcal{C}^\prime$. If $\mathcal{C}=\mathcal{C}^\prime$, an equivalence is said an \textit{automorphism} of $\mathcal{C}$. The set $\Aut(\mathcal{C})$ of all automorphisms of $\mathcal{C}$ is a group with respect to the product 
\begin{equation*}
    (f,\rho,g,h) \star (f',\rho',g',h')=(f' \circ f^{\rho'},\rho\rho',g^{\rho'} \circ g', f' \circ h^{\rho'} \circ g' + h')
\end{equation*}
where the composition is taken modulo $X^{q^{ns}}-X$.\\

\noindent To make this article as self-contained as possible, we shall present the state of the art regarding the known $(n,n,q;d)$-MRD codes in the linearized polynomials setting. Although, very recently, new constructions of both linear and non-linear rectangular MRD codes have been exhibited by Minjia Shi \textit{et al.}  in \cite{switching}.  We shall start from maximum rank distance code that are  $\Fq$-linear subspaces of $\tilde{\mathcal{L}}_{n,q,\sigma}$.\\

\noindent \textit{2.1 Linear MRD codes}\\

\noindent The first class of linear MRD codes, discovered by Delsarte \cite{Delsarte} and Gabidulin \cite{Gab0}, are known in literature as {\em Delsarte-Gabidulin codes}. Later in \cite{Gab}, Gabidulin and Kshevetskiy  provided a generalization of them, called {\em generalized Gabidulin codes} that we shall recap in our notations:
 let $ 1 \leq k \leq n$ be an integers, then a generalized Gabidulin code is the set of linearized polynomials \begin{equation*}
\mathcal{G}_{k,\sigma}=\left\{ \sum_{i=0}^{k-1} \alpha_i X^{\sigma^i} \colon \alpha_0,\ldots,\alpha_{k-1}\in \mathbb{F}_{q^n} \right\}
\end{equation*}
and it is an $(n,n,q;n-k+1)$-MRD code.

 In \cite{Shee}, Sheekey exhibited a wider class of linear MRD codes, called {\em twisted Gabidulin codes} and later generalized in \cite{LuTr} by Lunardon, Trombetti and Zhou. Let $n, k,h$ be positive integers, $1 \leq k<n$ and let $\eta$ be in $\mathbb{F}_{q^n}$ such that $\mathrm{N}_{q^{sn}/q^s}(\eta)\ne (-1)^{nk}$, then a \textit{generalized twisted Gabidulin code} is the set 
\begin{equation*}
\mathcal{H}_{k,\sigma}(\eta,h)=\left\{ \sum_{i=0}^{k-1} \alpha_i X^{\sigma^i} + \eta \alpha_0^{q^h}X^{\sigma^k}: \alpha_0,\ldots,\alpha_{k-1}\in \mathbb{F}_{q^n} \right\}
\end{equation*}
and it is an $(n,n,q;n-k+1)$-MRD code\footnote{The symbol $\mathrm{N}_{q^m/q^\ell}(x) = x^\frac{q^m-1}{q^\ell-1}$ denote the \textit{norm} of $x \in \mathbb{F}_{q^m}$ over $\mathbb{F}_{q^\ell}$, where $m,\ell$ are integers with $\ell \mid m$.}.

Puchinger \textit{et al.} provided a further generalization of generalized twisted Gabidulin codes in \cite{Puc}.\\
Finally, a family of maximum
rank distance codes in $\tilde{\mathcal{L}}_{n,q,\sigma}$, $n=2t$, $1 \leq k < n$ and $q$ odd was discovered by Trombetti and Zhou in 2019 and described  in \cite{trombetti_zhou}: let $\xi \in  \F_{q^n}$ satisfying that $\mathrm{N}_{q^n/q}(\xi )$ is a non-square in $\F_q$, then

\begin{equation*}
\mathcal{D}_{k,\sigma}(\xi )= \left\{
\sum_{i=0}^{k-1}\alpha_i X^{\sigma^i} +\xi \alpha_k X^{\sigma^{k}} : \alpha_0, \alpha_k \in  \F_{q^t}, \alpha_1,\ldots,\alpha_{k-1} \in \F_{q^n} \right \}
\end{equation*}
 is an $(n,n,q;n-k+1)$-MRD code.\\
\\
\noindent \textit{2.2 Additive MRD codes}\\

\noindent In \cite{Otal2}, Otal  and {\"O}zbudak  showed a family of non-linear additive MRD codes as a generalization of generalized twisted Gabidulin codes. Let $n, k, u, h$ be positive integers satisfying  $q=q_0^u$ and $1 \leq k<n$. Let $\eta \in \F_{q^n}$ such that $\mathrm{N}_{q^{sn}/q_0^s}(\eta) \ne (-1)^{nku}$. Then the set
\begin{equation}\label{OtalOzbudakadditive}
\mathcal{A}_{k,\sigma,q_0}(\eta,h) = \left\{ \sum_{i=0}^{k-1} \alpha_i X^{\sigma^i} + \eta \alpha_0^{q_0^h}X^{\sigma^k} : \alpha_0,\ldots,\alpha_{k-1} \in \mathbb{F}_{q^n} \right\}
\end{equation}
is an $\F_{q_0}$-linear $(n,n,q;n-k+1)$-MRD code. \\

Later, another family of non-linear additive MRD codes appeared in \cite{Sheekey3}. Let $F \in \mathbb{F}_q[Y]$ be irreducible polynomial of degree $r$ and let consider the ideal $\mathcal{I}=(F(X^{q^{sn}}))$ contained in $\mathcal{L}_{n,q,\sigma}$. Let $\rho \in \Aut(\F_q)$ with $\mathrm{Fix}(\rho)=\F_{q_0}$ and let $\eta \in \F_{q^n}$ such that $\mathrm{N}_{q^n/q_0}(\eta)\mathrm{N}_{q/q_0}((-1)^{kr(n-1)}F_0^k)\neq 1 $. Then,  the set
\begin{equation*}
\mathcal{S}_{k,\sigma}(\eta,\rho)=\left \{ \sum_{i=0}^{k-1}\alpha_iX^{\sigma^i}+ \eta \alpha_{0}^\rho X^{\sigma^k}+ \mathcal{I}: \alpha_0,\ldots,\alpha_{k-1} \in \F_{q^n} \right \} \subseteq 
\mathcal{L}_{n,q,\sigma}/\mathcal{I}\cong \F_{q^r}^{n \times n}
\end{equation*}
 is an $\F_{q_0}$-linear MRD code of $\F_{q^r}^{n \times n}$ of size $q^{nrk}$,  \cite[Theorem 7]{Sheekey3}. Note that for suitable choices of $F$, $\rho$ and $\eta$, the (twisted) generalized Gabidulin codes and the non-linear additive codes in \eqref{OtalOzbudakadditive} return, see \cite[Remark 3]{Sheekey3}.\\
 \\
\noindent \textit{2.3 Non-linear MRD codes}\\

\noindent In finite geometry $(n,n,q;n)$-MRD codes are known as {\em spread sets} (see e.g. \cite{Dem}) and there are examples for both cases linear and non-linear. These are linked to algebraic and geometric structures such as quasifields, semifields and translation planes coordinatized over them, see \cite{delaCruz} and \cite{handbook}.
The first class of non-linear MRD codes whose codewords are not invertible and so different from spread sets, are the $(3,3,q;2)$-MRD codes constructed by Cossidente \textit{et al.} \cite{Coss}. These arise from a geometrical context and have been generalized in two steps: firstly, by Durante and Siciliano \cite{Sic} obtaining a family of $(n,n,q;n-1)$-MRD codes, $n \geq 3$, and then by Donati and Durante obtaining a family of $(d+1,n,q;d)$-MRD codes, $2 \leq d \leq n-1$, see \cite{DonDur}.
As we shall extend this family further, we will describe it  in detail in Section \ref{Cfset}.

The second family of non-additive MRD codes for all $n,d$  has been constructed by Otal and {\"O}zbudak in \cite{Otal}: let $I$ be a subset of $\mathbb{F}_q$, $1 \leq k \leq n-1$ and consider
\begin{equation}\label{Turchi}
\mathcal{C}^{(1)}_{n,k,\sigma,I}=\left\{ \sum_{i=0}^{k-1} \alpha_i X^{\sigma^i}: \alpha_0,\ldots,\alpha_{k-1} \in \mathbb{F}_{q^n}, \mathrm{N}_{{q^n}/ q}(\alpha_0) \in I \right\},
\end{equation}
\[
\mathcal{C}^{(2)}_{n,k,\sigma,I}=\left\{ \sum_{i=1}^{k} \beta_i X^{\sigma^i}  : \beta_1,\ldots,\beta_{k} \in \mathbb{F}_{q^n}, \mathrm{N}_{{q^n}/ q}(\beta_k) \not \in (-1)^{n(k+1)}I \right\}.
\]
Then $\mathcal{C}_{n,k,\sigma,I}=\mathcal{C}^{(1)}_{n,k,\sigma,I} \cup \mathcal{C}^{(2)}_{n,k,\sigma,I} \subseteq \tilde{\mathcal{L}}_{n,q,\sigma}$ is an $(n,n,q;n-k+1)$-MRD code. In \cite[Corollary 2.1]{Otal}, they proved 
\begin{itemize}
\item[1.] if $q = 2$ or $I \in \{\emptyset,\{0\},\Fq^*,\Fq\}$ then $\mathcal{C}_{n,k,\sigma,I}$ is equivalent to a generalized Gabidulin code;
    \item[2.] if $q > 2$ and $I \not  \in \{\emptyset,\{0\},\Fq^*,\Fq\}$ , then $\mathcal{C}_{n,k,\sigma,I}$ is not an affine code (i.e. not a
translated version of an additive code).
\end{itemize}

In the following sections, we provide a geometric construction for a class of non-linear $(n,n,q;d)$-MRD $\mathcal{C}_{\sigma,T}$, $1 \in T \subseteq \mathbb{F}_q^*$,  $2 \leq d \leq n-1$, and puncturing properly a code in this relevant class, we get a code described in \cite{DonDur}. We shall show that this class is effectively new, i.e. any code is not equivalent to a non-linear code constructed by Otal and {\"O}zbudak.

\section{The geometric setting}\label{sec3}

Let $\mathbb{E}=\mathrm{End}_{\F_q}(\F_{q^n})$ be the vector space of all endomorphisms of $\F_{q^n}$ seen as a vector space over the field $\F_q$.  As  any element of  $\mathbb{E}$ corresponds  1-to-1 to a linearized polynomial 
$\alpha=\sum_{i=0}^{n-1}\alpha_iX^{\sigma^i} \in \tilde{\mathcal{L}}_{n,q,\sigma}$, we will define \textit{rank} of the polynomial $\alpha$  as the dimension  over $\F_q$ of $\mathrm{im}\, \alpha(x)$, where $\alpha(x)$ is the map  $x \in \F_{q^n} \mapsto \sum_{i=0}^{n-1}\alpha_i x^{\sigma^i} \in \F_{q^n}$,  in symbol $\rk \alpha = \dim_{\F_q} \mathrm{im} \,\alpha(x)$.
Also, the $\F_q$-linear map $\alpha(x)$ has rank $r$ if and only if the {\em Dickson matrix}
$$
D_\alpha=
\begin{pmatrix} 
\alpha_0 & \alpha_1 & \ldots & \alpha_{n-1} \\
\alpha_{n-1}^\sigma & \alpha_0^\sigma & \ldots & \alpha_{n-2}^\sigma \\
\vdots & \vdots & \vdots & \vdots\\
\alpha_1^{\sigma^{n-1}} & \alpha_{2}^{\sigma^{n-1}} & \ldots & \alpha_0^{\sigma^{n-1}}
\end{pmatrix}
$$
has rank $r$, see for more details \cite[Chapter 3]{Lidl}.

Let $V$ be a $v$-dimensional vector space over the field $\F_{q^n}$ and let $\PG(v - 1, q^n)=\PG(V, \F_{q^n} )$. A set $\Theta$ of points of $\PG(v-1, q^n)$ is an $\Fq$-\textit{linear set} if each point is defined by a non-zero vector of an $\F_q$-linear vector space $U \subset V$, in symbol
\begin{equation*}
    \Theta= \left \{ \langle  \textbf{u} \rangle_{\F_{q^n}}  \colon  \textbf{u} \in U, \textbf{u} \neq  \textbf{0}  \right \}.
\end{equation*}
If $\dim_{\F_q}U=u$, we say that $\Theta$  has \textit{rank} $u$. The size of $\Theta$ can be at most $\frac{q^u-1}{q-1}$ and if it is attained, $\Theta$ is said to be \textit{scattered}. If $u=v$ and $\langle \Theta \rangle=\PG(v-1,q^n)$, then $\Theta$ is a {\em (canonical) subgeometry} of $\PG(v-1,q^n)$. It follows that $\Theta$ is a canonical subgeometry if and only if any of its frame is also a frame of $\PG(v-1,q^n)$. 

Let $\Sigma \cong \PG(n-1,q)$ be a canonical subgeometry of $\PG(n-1,q^n)$ and let $\hat{\sigma}$ be a generator of the subgroup of $\pgammal(n,q^n)$  whose elements fixing pointwise $\Sigma$.
Let $\overline{S}=\PG(W,q^n)$ be a subspace of $\PG(n-1,q^n)$. The integer $\dim_{\F_{q^n}} W=w $ will be called the \textit{rank} of $\overline{S}$. Then $S=\overline{S} \cap \Sigma$ is a subspace of $\Sigma$ of rank at most $w$. We will say that $\overline{S}$ is a {\em subspace of} $\Sigma$ if $S$ and $\overline{S}$ have the same rank. In particular, this holds if and only if $\overline{S}$ is fixed by the collineation $\hat{\sigma}$ (see e.g. \cite{Lun1}).
Any point $P$ of $\PG(n - 1, q^n)$ defines the subspace $$L_{P,\hat{\sigma}} = \langle P, P^{\hat{\sigma}}, \ldots, P^{\hat{\sigma}^{n-1}}
\rangle.$$
If the collineation $\hat{\sigma}$ is clear from the context, we will indicate $L_{P,\hat{\sigma}}$ simply by $L_{P}$.
Note that $L_{P}$ has rank at most $n$ and it is fixed by $\hat{\sigma}$, and so $L_{P}$ is a subspace of $\Sigma$.
Denote by $(X_0,X_1,\ldots,X_{n-1})$ the homogeneous projective coordinates of $\PG(n-1,q^n)$ and let $P(\alpha_0,\alpha_1,\ldots, \alpha_{n-1})$ be a point of $\PG(n-1,q^n)$. It will be of {\em type $r$ with respect to $\Sigma$} if $L_P$ has rank $r$. In particular, if $\hat{\sigma}$ is the collineation of $\PG(n-1,q^n)$ defined by $$(X_0,X_1,\ldots,X_{n-1})^{\hat{\sigma}}=(X_{n-1}^{\sigma},X_0^{\sigma},\ldots,X_{n-2}^{\sigma}).$$
Then, the collineation $\hat{\sigma}$ fixes pointwise the canonical subgeometry
\begin{equation}\label{canonicalsub}
    \Sigma_{n,n}=\{  (x,x^{\sigma},\ldots,x^{\sigma^{n-1}})  : x \in \F^*_{q^n} \} \cong \PG(n-1,q).
\end{equation}
and the point $P(\alpha_0,\alpha_1,\ldots, \alpha_{n-1})$ is of type $r$ with respect to $\Sigma_{n,n}$ if and only $\alpha=\sum^{n-1}_{i=0}\alpha_i X^{\sigma^i}$   has rank $r$, or also the Dickson matrix $D_{\alpha}$ has rank $r$.\\

Let $\PG(\mathbb{F}_q^{m \times n},\F_q)=\PG(mn-1,q)$, $m \leq n$, and let $\mathcal{S}_{m,n}$ be the {\em Segre variety} of $\PG(\mathbb{F}_q^{m \times n},\F_q)$, i.e., $\mathcal{S}_{m,n}$ is the set of all points $\langle A \rangle_{\F_q}$ in $\PG(\mathbb{F}_q^{m \times n},\F_q)$ such that $\rk A=1$, see \cite[Section 4.5]{HirsThas}. 
This can be seen as the $\F_q$-field reduction of the set of points
$$\Sigma_{m,n}=\{(x,x^{\sigma},\ldots,x^{\sigma^{m-1}})  \colon x \in \F^*_{q^n} \} \cong \PG(n-1,q)$$
of $\PG(m-1,q^n)$ in $\PG(mn-1,q)$, see \cite[Section 2.2]{Lav}.\\

Let $\mathcal{A}$ be a subset of $\PG(n-1,q)$ and denote by $\Omega_h(\mathcal{A})$ the \textit{$h$-secant variety} of $\mathcal{A}$, i.e. the union of the $\ell$-dimensional projective subspaces spanned by points of $\mathcal{A}$ for any $0 \leq \ell \leq  h $,  \cite{Har}. Note that $\Omega_0(\mathcal{A})=\mathcal{A}$ and for any $1 \leq h \leq n-1$, $\Omega_{h-1}(\mathcal{A}) \subseteq \Omega_{h}(\mathcal{A})$. Moreover, if $\mathcal{A} \subset \PG(n-1,q)$ such that $\langle \mathcal{A} \rangle =\PG(t-1,q) \subset \PG(n-1,q)$, then $\Omega_{h}(\mathcal{A})=\Omega_{t-1}(\mathcal{A})$ for any $ t \leq h \leq n-1 $. A set of points $\mathcal{E} \subset \PG(n-1,q)$ is called an \textit{exterior set} with respect to $\Omega_h(\mathcal{A})$ if any line joining two points of $\mathcal{E}$ is disjoint from $\Omega_h(\mathcal{A})$. The following theorem proves an upper bound for the size of an exterior set.
\begin{thm}\label{size_ext}
Let $\mathcal{A}\subset \PG(n-1,q)$ such that $\langle \mathcal{A} \rangle = \PG(n-1,q)$. Let $\mathcal{E} \subset \PG(n-1,q)$ be an exterior set with respect to $\Omega_h(\mathcal{A})$, $0 \leq h \leq n-1$. Then
   $$
   |\mathcal{E}|\leq \frac{q^{n-h-1}-1}{q-1}.
   $$
\end{thm}
\begin{proof}
Firstly, we may assume $h\leq n-3$. Indeed, if $h \in \{n-2,n-1\}$ the claim is trivially satisfied: if $h=n-1$, then $\Omega_{n-1}(\mathcal{\mathcal{A}})=\PG(n-1,q)$ and $\mathcal{E}=\emptyset$; if $h=n-2$, then $\Omega_{n-2}(\mathcal{\mathcal{A}})$ contains hyperplanes and $|\mathcal{E}| \leq 1$. \\
Now, let $h \leq n-3$ and let $S_h$ be an $h$-dimensional subspace contained in $\Omega_h(\mathcal{A})$. The number of $(h+1)$-dimensional projective subspaces of $\PG(n-1,q)$ containing $S_h$ equals the number of points of an $(n-h-2)$-dimensional projective subspace of $\PG(n-1,q)$. Since $\mathcal{E}$ is an exterior set, then it meets any $(h+1)$-dimensional projective subspace containing $S_h$ in at most one point. This proves the claim.
\end{proof}

Given $M, N$ two sets of points of $\PG(n- 1, q)$, with $M \cap N =\emptyset$, we will denote by $\mathcal{K}(M, N)$ the {\em cone with vertex $M$ and base $N$}, i.e. $\mathcal{K}(M,N)$ is the set of all points belonging to a line joining a point of $M$ and a point of $N$.\\
Let suppose  that $M,N$  are two disjoint subsets of $\PG(n-1,q)$ such that $\langle M \rangle$ and $\langle N \rangle$ are disjoint subspaces having rank at least 2. Let $P,Q$ be two distinct points in $M$ and $P',Q'$ be two distinct points in $N$. It is straightforward to show that the lines $PP'$ and $QQ'$ are disjoint. Indeed, if $PP' \cap QQ' \ne \emptyset$ then the plane spanned by $PP'$ and $QQ'$ meets the subspace $\langle M\rangle$ in the line $PQ$ and the subspace $\langle N\rangle$ in the line $P'Q'$, a contradiction. Then, it follows that
\begin{equation*}
|\mathcal{K}(M,N)|=|M|+|N|+(q-1)\cdot |M| \cdot |N|.
\end{equation*}
If the subspaces $\langle M \rangle$ and $\langle N \rangle$ of $\PG(n-1,q)$ are not disjoint, then the lines $PP'$ and $QQ'$ can be incident and, clearly,
\begin{equation*}
|\mathcal{K}(M,N)|\leq |M|+|N|+(q-1)\cdot |M| \cdot |N|.
\end{equation*}

\begin{corollary}\label{size_ext_subsp}
Let $\mathcal{A}\subset \PG(n-1,q)$ such that $\langle \mathcal{A} \rangle = \PG(t-1,q)$, $1 \leq t < n$, and let $\mathcal{E} \subset \PG(n-1,q)$ be an exterior set with respect to $\Omega_h(\mathcal{A})$, $0 \leq h \leq n-1$. Then $\mathcal{E}$ is contained in a cone $\mathcal{K}=\mathcal{K}(S_{n-t-1},\bar{\mathcal{E}}),$ with base $\bar{\mathcal{E}}=\mathcal{E} \cap \langle \mathcal{A} \rangle$ and vertex  an $(n-t-1)$-dimensional subspace $S_{n-t-1}$ complementary with $\langle \mathcal{A} \rangle$. Moreover,
\begin{equation}
   |\mathcal{E}|\leq 
   \begin{cases}
       \frac{q^{n-h-1}-1}{q-1} & \textnormal{if} \quad 0 \leq h \leq t-1,\\
      \frac{q^{n-t}-1}{q-1} & \textnormal{otherwise}.
   \end{cases}
   \end{equation}
\end{corollary}
\proof
Let suppose $0 \leq h \leq t-1$. Clearly,  if $\bar{\mathcal{E}}\subset \langle \mathcal{A} \rangle$ is an exterior set with respect to $\Omega_h(\mathcal{A}) \subseteq \langle \mathcal{A} \rangle$, $0 \leq h \leq t-1$, by Theorem \ref{size_ext},
    $$
    |\bar{\mathcal{E}}| \leq \frac{q^{t-h-1}-1}{q-1}.
    $$ 
Let $S_{n-t-1}$ be an $(n-t-1)$-dimensional subspace of $\PG(n-1,q)$ complementary with $ \langle \mathcal{A} \rangle$. Then,
\begin{eqnarray*}
     |\mathcal{K}(S_{n-t-1},\bar{\mathcal{E}})|&=&|\bar{\mathcal{E}}| \cdot |S_{n-t-1}| \cdot (q-1)+ |\bar{\mathcal{E}}|+|S_{n-t-1}| \\
    &\leq& \frac{q^{t-h-1}-1}{q-1} \cdot \frac{q^{n-t}-1}{q-1} \cdot (q-1) + \frac{q^{t-h-1}-1}{q-1} + \frac{q^{n-t}-1}{q-1} \\
    &=& \frac{q^{n-h-1}-1}{q-1}.
\end{eqnarray*}
If $ t \leq h \leq n-1$, since $\Omega_h(\mathcal{A})=\Omega_{t-1}(\mathcal{A})$ and any exterior set contained in $\langle \mathcal{A} \rangle$ with respect $\Omega_h(\mathcal{A})$ is the empty set, the claim follows.
\endproof

An exterior set $\mathcal{E} \subset \PG(n-1,q)$ with respect to $\Omega_h(\mathcal{A})$, $\langle \mathcal{A} \rangle= \PG(t-1,q)$, $1 \leq h+1 \leq t \leq n$, is called \textit{maximum} if its size $\vert \mathcal{E} \vert =\frac{q^{n -h-1}-1}{q-1}$.
Note that the image of $\Omega_h(\Sigma_{m,n}) \subset \PG(m-1,q^n)$ under the $\F_q$-field reduction is the $h$-secant variety $\Omega_h(\mathcal{S}_{m,n})$ of the points whose the representative matrices in $\mathbb{F}_q^{m \times n}$ have rank at most $h+1$. The (maximum) exterior sets with respect to $\Omega_h(\Sigma_{m,n})$ are related to (maximum) rank distance codes. More precisely,
\begin{thm}\label{code}
    Let $\mathcal{E}$ be an exterior set with respect to $\Omega_h(\Sigma_{m,n})$ of $\PG(m-1,q^n)$ and denote by $\mathcal{E}^{\prime}$ the image of $\mathcal{E}$ under the $\F_q$-field reduction. Then, the set 
    \begin{equation}\label{codice}
    \mathcal{C}=\left\{ \rho M : \langle M \rangle_{\mathbb{F}_{q}}  \in \mathcal{E}^{\prime}, \rho \in \mathbb{F}_{q}\right\}
    \end{equation}
    is an $(m, n,q; h+2)$-RD code closed under $\F_q$-multiplication. In addition, if $\mathcal{E}$ is maximum then $\cC$ is an MRD.
\end{thm}

Note that the $(m,n,q;h+2)$-RD codes closed under $\F_q$-multiplication and the exterior sets with respect to $\Omega_h(\mathcal{S}_{m,n})$ of $\PG(mn-1,q)$ are in 1-to-1 correspondence. Since the pre-image under the field reduction of one of these sets may not be an exterior set of $\PG(m-1,q^n)$ with respect to $\Omega_h(\Sigma_{m,n})$, then there exist rank distance codes that don't arise from exterior sets with respect to $\Omega_h(\Sigma_{m,n})$ of $\PG(m-1,q^n)$.
See \cite{Coss,Coop, DonDur} and the references therein for more details on exterior sets.

 \section{Non-linear MRD codes arising from $C_F^{\sigma}$-set}\label{Cfset}

In \cite{DonDur}, Donati and Durante  exhibited a class of non-linear MRD codes with parameters $(d+1,n,q;d)$ with $q >2$, $n \geq 3$  and $ 2\leq d \leq n-1$ by means of a family of maximum exterior sets with respect the $(d-2)$-secant variety of a subgeometry isomorphic to $\PG(d,q)$ in $\Sigma_{d+1,n}$. Now, we recall briefly their construction. So, let $\PG(d,q^n)$ be a $d$-dimensional projective space and consider $A$ and $B$ two distinct points of $\PG(d,q^n)$,  $\mathcal{S}_A$ and $\mathcal{S}_B$ be the stars of lines (pencils of lines if $d=2$) through $A$ and $B$, respectively. Let $\sigma$ be the Frobenius automorphism of $\mathbb{F}_{q^n}$ defined by $ x \mapsto x^{q^s}$ with $\gcd(n,s)=1$ and consider $\Phi$ a $\sigma$-collineation between $\mathcal{S}_A$ and $\mathcal{S}_B$ which does not map the line
$AB$ into itself and such that the subspace spanned by the lines $\Phi^{-1}(AB), AB, \Phi(AB)$ has
dimension $\min\{3, d\}$. The set $\mathcal{X}$ of points of intersection of corresponding lines under the collineation $\Phi$ is called {\em $C_F^{\sigma}$-set} of $\PG(d,q^n)$ and the points $A$ and $B$ are called the {\em vertices} of $\mathcal{X}$, see \cite{DonDur0,DonDur}. 
There, it is proved that every $C_F^{\sigma}$-set $\mathcal{X}$ of $\PG(d,q^n)$ with vertices $A$ and $B$ is the union of $\{A,B \}$ and $q-1$ pairwise disjoint subsets, called the {\em components} of $\mathcal{X}$, each of which is a scattered $\mathbb{F}_q$-linear set of rank $n$. For more details on linear sets, see e.g. \cite{Polve}.

Let $N_a = \{y \in \mathbb{F}_{q^n}: \mathrm{N}_{q^n/q}(y) = a\}$ for any $a \in \mathbb{F}_q^{\ast}$. Without loss of generality, we may assume that $A=(0,\ldots,0,1)$ and $B=(1,0,\ldots,0)$ and then the sets 
\[
\mathcal{X}_a=\left\{(1,t,t^{\sigma+1},\ldots,t^{\sigma^{d-1}+\ldots+\sigma+1}): t \in N_a \right\},
\]
with $a \in \mathbb{F}_q^{\ast}$, are the components of $\mathcal{X}$. Since every $\mathcal{X}_a$ has $(q^n-1)/(q-1)$ points, then $\mathcal{X}$ has $q^n+1$ points. In particular, every $\mathcal{X}_a$ is isomorphic to $\PG(n-1,q)$, see \cite[Remark 3.5]{DonDur}.
For any $a \in \mathbb{F}_q^{\ast}$, the line $AB$ of $\PG(d,q^n)$ is partitioned into $\{ A, B\}$ and the $q-1$ sets
\[
J_a=\left \{ (1,0,\ldots,0,(-1)^{d+1}t): t \in N_a \right\},
\]
each of which
is a scattered $\mathbb{F}_q$-linear set (of pseudoregulus type with transversal points $A$ and $B$), see \cite[Remark 2.2]{DonDur0} and \cite{Lun}.
Note that $\mathcal{X}_1=\Sigma_{d+1,n}$ and let $\Pi$ be a subgeometry of $\Sigma_{d+1,n}$ isomorphic to $\PG(d,q)$. Then for any $T \subseteq \mathbb{F}_q^{\ast}$, $1 \in T$, the set
\[
\mathcal{E}=\Bigg( \mathcal{X} \setminus \bigcup_{a \in T} \mathcal{X}_a \Bigg) \cup \bigcup_{a \in T} J_a
\]
is a maximum exterior set with respect to $\Omega_{d-2}(\Pi)$, see \cite[Theorem 5.1]{DonDur}. Hence, the set $\mathcal{E}$ corresponds a $(d+1,n,q;d)$-MRD code.

\section{Embeddings and cones over  exterior sets}\label{sec5}

Let $\Sigma \cong \PG(n-1,q)$ be a canonical subgeometry of $\PG(n-1,q^n)$ and consider a subspace $\Lambda^{\star}$ of rank $k$ disjoint from $\Sigma$ and $\Lambda$ a subspace of $\PG(n - 1, q^n)$ of rank
$n - k$ disjoint from $\Lambda^{\star}$. Let $\Gamma$ be the projection of $\Sigma$ from $\Lambda^{\star}$ to $\Lambda$, i.e.,
$$
\Gamma=\mathrm{p}_{\Lambda^{\star},\Lambda}(\Sigma)=\left \{\langle \Lambda^{\star}, P \rangle \cap \Lambda : P \in \Sigma    \right\}.
$$
We recall the following definition given in \cite{Lun2}.

\begin{defn}
\textnormal{Let $\Gamma=\mathrm{p}_{\Lambda^{\star},\Lambda}(\Sigma)$ be an $\F_q$-linear set of rank $n$. It is called {\em $(n-k-1)$-embedding of $\Sigma$} if any subspace of $\Sigma$ of rank $n-k$ is disjoint from $\Lambda^{\star}$.}
\end{defn}

Note that $\Gamma=\mathrm{p}_{\Lambda^\star,\Lambda}(\Sigma)$ is an $(n-k-1)$-embedding if and only if for any choice of $n-k$ independent points $R_1,R_2, \ldots, R_{n-k}$  of $\Sigma$,
\[\Lambda=\langle R'_1,R'_2,\ldots,R'_{n-k} \rangle \]
where $R'_i=\mathrm{p}_{\Lambda^\star,\Lambda}(R_i)$, $i=1,2,\ldots,n-k$, also this is equivalent to say that  $$\mathrm{p}_{\Lambda^\star,\Lambda}:  P\in  \Sigma \longrightarrow \langle P,\Lambda^\star \rangle \cap \Lambda \in \Lambda$$ induces an injective
map from the set of all subspaces of $\Sigma$ of rank $\ell \leq n - k-1 $ to the set of all subspaces of $\Lambda$ of the same rank $\ell$ and, hence, if $R \in \Omega_{n-k-1}(\Sigma)$ then $\langle R, \Lambda^\star\rangle \cap \Lambda \in \Omega_{n-k-1}(\Gamma)$.
\begin{remark}
\textnormal{If $\Gamma=\mathrm{p}_{\Lambda^*,\Lambda}(\Sigma)$  is an $(n-k-1)$-embedding of $\Sigma$, it is not true that $\Omega_{n-k-1}(\Gamma) \not \subseteq \Omega_{n-k-1}(\Sigma) \cap \Lambda$. For instance, let $$\Sigma_{5,5}=\{(x,x^{\sigma},x^{\sigma^2},x^{\sigma^3},x^{\sigma^4}):x \in \mathbb{F}_{q^5}^*\}$$ be a fixed canonical subgeometry of $\PG(4,q^5)$ and let $\Lambda^\star(0,0,0,0,1)$ and $\Lambda:X_4=0$ be a point and a hyperplane of $\PG(4,q^5)$, respectively. Clearly, $\Omega_1(\Sigma)$ and $\Omega_1(\Gamma)$ contains lines joining two distinct points of $\Gamma$ and $\Sigma$, respectively. If $\Omega_{1}(\Sigma) \cap \Lambda$ containes a line then it would contain two points of $\Sigma$. Since  $\Lambda$ and $\Sigma$ are disjoint, we get a contradiction.}
\end{remark}

\begin{lem}\cite[Theorem 6]{Lun2}\label{rank}
The $\F_q$-linear set $\Gamma=\mathrm{p}_{\Lambda^{\star},\Lambda}(\Sigma)$ is an $(n-k-1)$-embedding of $\Sigma$ if and only if any point $P \in \Lambda^{\star}$ is of type $\geq n-k+1$.
\end{lem}

\proof
By definition $\Gamma$ is an $(n - k - 1)
$-embedding if and only if $\Lambda^\star$ is disjoint from any
subspace of $\Sigma$ of rank $n - k$. This is equivalent to saying that for any point $P$ of $\Lambda^\star$ the
subspace $L_P$ has rank at least $n - k + 1$, i.e. the point $P$ is of type at least $n - k + 1$. 
\endproof

In \cite{Lun2}, in order to get  a  linear MRD-code, a geometric property for the cone having vertex $\Lambda^\star$ and base a linear set of rank $n$  has been shown. More precisely,

\begin{thm}\cite[Theorem 5]{Lun2}
Let $\Sigma \cong \PG(n - 1, q)$ be a canonical subgeometry of $\PG(n - 1, q^n)$ and let $\Lambda^{\star}$ and $\Lambda$ be subspaces of $\PG(n-1,q^n)$ of rank $k$ and $n-k$, respectively, such that $\Lambda^\star \cap \Sigma= \emptyset= \Lambda^\star \cap \Lambda$.\\
Let $\Gamma = \mathrm{p}_{\Lambda^{\star},\Lambda}(\Sigma)$ be an $(n - k-1)$-embedding of $\Sigma$ and let $\Theta \subset \Lambda$ be an $\mathbb{F}_q$-linear set of rank $n$. If no points of $\Theta$ belong to a subspace spanned by at most
$n -k$ points of $\Gamma$, then any point of $\mathcal{K}=\mathcal{K}(\Lambda^{\star}, \Theta)$ is of type $\geq n - k + 1$. \\ Moreover, if $\Sigma=\Sigma_{n,n}$ as in \eqref{canonicalsub},  the set  $\mathcal{C} =\left\{ \sum_{i=0}^{n-1}\alpha_i X^{\sigma^i}: (\alpha_0,\alpha_1,\ldots,\alpha_{n-1}) \in \mathcal{K} \right\} $ is a linear MRD-code with minimum
distance $d = n - k + 1$.
\end{thm}

Here, we shall adapt the result above to the case of cone such that its base is not necessarily a linear set.

\begin{thm}\label{embedd}
Let $\Sigma \cong \PG(n - 1, q)$ be a canonical subgeometry of $\PG(n - 1, q^n)$ and let $\Lambda^{\star}$ and $\Lambda$ be subspaces of $\PG(n-1,q^n)$ of rank $k-2$ and $n - k+2$, respectively, such that $\Lambda^\star \cap \Sigma= \emptyset= \Lambda^\star \cap \Lambda$. Let $\Gamma = \mathrm{p}_{\Lambda^{\star},\Lambda}(\Sigma)$ be an $(n - k+1)$-embedding of $\Sigma$ and let $\mathcal{E}$ be a (maximum) exterior set with respect to $\Omega_{n-k-1}(\Gamma)$.\\
Then, $\mathcal{K}=\mathcal{K}(\Lambda^{\star},\mathcal{E})$ is a (maximum) exterior set with respect to $\Omega_{n-k-1}(\Sigma)$.\\
\end{thm}

\begin{proof} Let $P,Q$ be two points in $\mathcal{K}=\mathcal{K}(\Lambda^\star,\mathcal{E})$. We shall distinguish some cases:
\begin{itemize}
    \item [-]  $P,Q \in \Lambda^\star$. Since $\Gamma$ is an $(n-k+1)$-embedding the line $PQ$ is disjoint from  the subspaces of rank $(n-k)$ of $\Sigma$ and so from $\Omega_{n-k-1}(\Sigma)$.
    \item [-] $P,Q \in \mathcal{E}$. Suppose that the line $PQ$ meets a subspace $S$ of $\Sigma$ of rank $(n-k)$ and consider 
   \begin{equation}
    \emptyset \not = \langle \Lambda^\star, S \cap PQ  \rangle \cap \Lambda \subseteq \langle \Lambda^\star, S \rangle \cap \Lambda \cap PQ.
    \end{equation}
    This is a contradiction by the hypothesis done over $\mathcal{E}$ and since $\langle \Lambda^\star,S \rangle \cap \Lambda$ is a subspace of rank $n-k$.
    \item [-] the line $PQ \subseteq \mathcal{K} \setminus (\Lambda^\star \cup \mathcal{E})$ and joins a point of $\Lambda^\star$ and a point of $\mathcal{E}$. Then, without loss of generality, we may suppose that $P \in \Lambda^\star$ and $Q \in \mathcal{E}$. 
    If the line $PQ$ meets a subspace $S$ of $\Sigma$ of rank $(n-k)$, then
    \begin{equation}
   Q \in \langle \Lambda^\star, S \cap PQ  \rangle \cap \Lambda \subseteq \langle \Lambda^\star, S \rangle \cap \Lambda.
    \end{equation}
Then $Q$ belongs to a space spanned by points of $\Gamma$ with rank $(n-k)$, a contradiction.
\item [-] the line $PQ \subseteq \mathcal{K}  \setminus (\Lambda^\star \cup \mathcal{E}) $ and does not joint a point of $\Lambda^\star$ and a point of $\mathcal{E}$ or  $PQ \not \subseteq \mathcal{K}$. If $PQ$ meets a subspace $S$ of $\Sigma$ of rank $(n-k)$, then
   \begin{equation}
  \emptyset \not= \langle \Lambda^\star, S \cap PQ  \rangle \cap \Lambda \subseteq \langle \Lambda^\star, S \rangle \cap \langle \Lambda^\star, PQ \rangle \cap \Lambda.
    \end{equation}
    Since the projection from $\Lambda^\star$ to $\Lambda$ of the line $PQ$ is a line through two points $P',Q'$ of $\mathcal{E}$, we get that this line meets the space $\langle \Lambda^\star, S \rangle \cap \Lambda$ spanned by $(n-k)$ points of $\Gamma$, a contradiction.
\end{itemize}
Then we have showed that any line through two points of $\mathcal{K}$ is disjoint from $\Omega_{n-k-1}(\Sigma)$. Now, since $\langle \Gamma \rangle = \Lambda$, if $\mathcal{E}$ is a maximum exterior set with respect to $\Omega_{n-k-1}(\Gamma)$, we get
\begin{eqnarray}
|\mathcal{K}(\Lambda^\star,\mathcal{E} )|&=& |\Lambda^\star|+|\mathcal{E} |+|\Lambda^\star||\mathcal{E} |(q^n-1) \nonumber \\
&=& \frac{q^{n(k-2)}-1}{q^n-1}+\frac{q^{2n}-1}{q^n-1}+\frac{q^{n(k-2)}-1}{q^n-1}\cdot \frac{q^{2n}-1}{q^n-1}(q^n-1) \nonumber \\
&=& \frac{q^{2n}-1}{q^n-1}+\frac{q^{n(k-2)}-1}{q^n-1}q^{2n} \nonumber \\
&=& \frac{q^{nk}-1}{q^n-1}. \nonumber  \end{eqnarray}
\end{proof}

In the same setting of theorem above, it seems that if $\mathcal{K}(\Lambda^\star,\mathcal{E})$ is an exterior set with respect to $\Omega_{n-k-1}(\Sigma)$ where $\mathcal{E}$ is a subset contained in $\Lambda$, then $\mathcal{E}$ may not be an exterior set with respect $\Omega_{n-k-1}(\Gamma)$. If $\Omega_{n-k-1}(\Sigma) \cap \Lambda=\Omega_{n-k-1}(\Gamma)$, this turns to be true. Indeed, if $A,B \in \mathcal{E}$ and there exists $P \in AB \cap \Omega_{n-k-1}(\Gamma)$, then $P \in AB \cap \Omega_{n-k-1}(\Sigma)$, a contradiction.

\section{The class of non-linear MRD codes $\cC_{\sigma,T}$}\label{sec6}

In this section, as application of Theorem \ref{embedd}, by a suitable choice of $\Lambda^\star$, $\Lambda$, an $(n-k+1)$-embedding $\Gamma$ of $\Sigma=\Sigma_{n,n}$, $2 \leq k \leq n-1$, and a maximum exterior set with respect to $\Omega_{n-k-1}(\Gamma)$ in $\PG(n-1,q^n)$, we will get a class of non-linear MRD codes in $\tilde{\mathcal{L}}_{n,q,\sigma}$.
Precisely,
let 
 $$ \Lambda^\star :X_{0}=X_1=\ldots=X_{n-k+1}=0
$$ 
and $$\Lambda:X_{n-k+2}=X_{n-k+3}=\ldots=X_{n-1}=0
$$
be disjoint subspaces of rank $(k-2)$ and $(n-k+2)$ in $\PG(n-1,q^n)$, respectively. Consider the linear set of rank $n$
\begin{equation}
\Gamma=\mathrm{p}_{\Lambda^\star,\Lambda}(\Sigma)=\{(x,x^{\sigma},\ldots,x^{\sigma^{n-k+1}},0,\ldots,0) : x \in \F^*_{q^n} \}
\end{equation}
where $\Sigma$ is the canonical geometry as in \eqref{canonicalsub} and let $P=(0,\ldots,0,\beta_{n-k+2},\ldots,\beta_{n-1}) \in \Lambda^\star$. This is a point of type at least $n-k+3$. Indeed, the linearized polynomial $$\beta=\beta_{n-k+2}X^{\sigma^{n-k+2}}+\ldots+\beta_{n-1}X^{\sigma^{n-1}}$$
has rank at least $n-k+3$ because the polynomial 
 $\beta \circ X^{\sigma^{k-2}}=\beta_{n-k}X+\ldots+\beta_{n-1}X^{\sigma^{k-3}}$
has at most $q^{k-3}$ roots and $\rk (\beta \circ X^{\sigma^{k-2}})= \rk \beta$. So, by Lemma \ref{rank}, we get that $\Gamma$ is an $(n-k+1)$-embedding of $\Sigma$.
Finally, let
$$A=(\underbrace{0,\ldots,0}_{n-k+1},1,0,\ldots,0) \textnormal{ and }B= (1,0,\ldots,0) $$ be points in $\Lambda$ and consider  $$\mathcal{X}=\bigcup_{a \in \Fq^{\ast}} \mathcal{X}_a \cup \, \{A,B\},$$where \begin{equation*}
\mathcal{X}_a=\left\{ (1,t,t^{\sigma+1},\ldots,t^{\sigma^{n-k}+\ldots+\sigma+1},0,\ldots,0): t \in N_a \right\}. 
\end{equation*}
The set $\mathcal{X}$ is  the $C_F^{\sigma}$-set with vertices $A$ and $B$ generated by a $\sigma$-collineation $\Phi$ between the star of lines through $A$ and $B$ contained in $\Lambda$. Let  $T \subseteq \mathbb{F}_q^{\ast}$, $1 \in T$, then the set
\begin{equation}\label{DonDurset}
\mathcal{E}=\Bigg( \mathcal{X} \setminus \bigcup_{a \in T} \mathcal{X}_a \Bigg) \cup \bigcup_{a \in T} J_a
\end{equation}
with
\begin{equation*}
J_a=\left\{ (1,\underbrace{0,\ldots,0}_{n-k}, (-1)^{n-k}t,0,\ldots,0): t \in N_a \right\}, \quad a \in \F^*_q.
\end{equation*}
is a maximum exterior set with respect to $\Omega_{n-k-1}(\Gamma)$ as proved in \cite[Theorem 5.1]{DonDur} of size $q^n+1$. Therefore, the hypothesis of Theorem \ref{embedd} are satisfied and the cone
$\mathcal{K}(\Lambda^\star,\mathcal{E})$ is a maximum exterior set with respect to $\Omega_{n-k-1}(\Sigma)$ and by Theorem \ref{code}, the set $\cC_{\sigma,T}  \subset \F_{q}^{n \times n}$, as in \eqref{codice}, is a non-linear MRD code with minimum distance $d=n-k+1$ which cannot be a translated version of an additive code. \\
Note that the punctured code $\cC^{[k-2]}_{\sigma,T} \subset  \F^{(n-k+2) \times n}_{q}$ obtained from $\cC_{\sigma,T}$ by deleting the last $(k-2)$ rows is exactly the code constructed in \cite[Theorem 5.1]{DonDur}, while for $k=2$ the code appeared in \cite{Sic} and for $n=3,k=2$ in \cite{Coss}.

By Theorem \ref{code} and Theorem \ref{embedd}, the $\cC_{\sigma,T}$ is a subset of square matrices of order $n$ and so, it can be seen as a subset of $\tilde{\mathcal{L}}_{n,q,\sigma}$. Indeed, fixing $T \subseteq \F^*_{q}$, $1 \in T$, the homogeneous coordinates of the points belonging to $\mathcal{E}$ have one of the following forms 
\begin{equation}
\begin{split}
&
(\alpha,\alpha^\sigma\xi,\alpha^{\sigma^2}\xi^{\sigma+1},\ldots,\alpha^{\sigma^{n-k+1}}\xi^{\sigma^{n-k}+\ldots+1},0,\ldots,0)\\
&(\alpha,\underbrace{0,\ldots,0}_{n-k}, (-1)^{n-k}\alpha^\sigma\eta,0,\ldots,0)\\
&
A=(\underbrace{0,\ldots,0}_{n-k+1},\alpha,0,\ldots,0)
\textnormal{ and }
B=(\alpha,0,\ldots,0)
\end{split} 
\end{equation}
with $\mathrm{N}_{q^n/q}(\xi) \in \F_{q}^*\setminus T$ and $\mathrm{N}_{q^n/q}(\eta) \in T$. So the non-linear $(n,n,q;d)$, $d=n-k+1$, MRD code $\cC_{\sigma,T}$ is the set of $\sigma$-linearized polynomials with $\sigma$-degree at most $n-1$  whose coefficients are the homogeneous coordinates of a point in $\mathcal{K}(\Lambda^\star,\mathcal{E})$ together with the zero map, i.e.

\begin{equation}\label{codepol}
\begin{split}
\mathcal{C}_{\sigma,T}&=\left\{ \sum_{i=0}^{n-k+1} \lambda \alpha^{\sigma^i} \xi^\frac{\sigma^i-1}{\sigma-1}X^{\sigma^i}+\sum_{i=n-k+2}^{n-1}\beta_i X^{\sigma^i} : \lambda, \alpha, \beta_i \in \mathbb{F}_{q^n}, \mathrm{N}_{q^n/q}(\xi) \in \mathbb{F}_q^{\ast} \setminus T \right\}  \\
&\cup \left\{ \lambda \alpha X +(-1)^{n-k+2} \lambda \alpha^{\sigma}\eta X^{\sigma^{n-k+1}}+\sum_{i=n-k+2}^{n-1}\beta_i X^{\sigma^i} : \lambda, \alpha, \beta_i \in \mathbb{F}_{q^n}, \mathrm{N}_{q^n/q}(\eta) \in T  \right\} \\
&\cup \left\{ \alpha X^{\sigma^{n-k+1}}+\sum_{i=n-k+2}^{n-1}\beta_i X^{\sigma^i} : \alpha, \beta_i \in \mathbb{F}_{q^n} \right\} \cup \left\{ \alpha X+\sum_{i=n-k+2}^{n-1}\beta_i X^{\sigma^i} : \alpha, \beta_i \in \mathbb{F}_{q^n} \right\} .
\end{split}
\end{equation}

\section{The equivalence issue for $\cC_{\sigma,T}$}\label{sec7}

Finally, in order to state the novelty of the class of non-linear codes obtained, we have to compare a code  of type $\cC_{\sigma,T}$, $1 \in T \subseteq \F_q^*$, with the only known non-linear MRD codes made up of square matrices: the codes of type  $\cC_{n,k,\sigma,I}$ with $I \subseteq \F_q$, cf. \eqref{Turchi}.\\
It easy straightforward to see that the code $\cC_{\sigma,T}$ in \eqref{codepol} is closed under $\F_{q^n}$-multiplication and hence under $\F_q$-multiplication. While, the code $\cC_{n,k,\sigma,I}$ in \eqref{Turchi}, if $q>2$ and $I \not \in \{\emptyset,\{0\},\F_q^*,\F_q\}$  is not closed under $\F_q$-multiplication. Therefore, by Theorem \ref{code}, $\cC_{n,k,\sigma,I}$ cannot arise from an exterior set $\mathcal{E}'$ of $\PG(n^2-1,q)$ with respect to the secant variety $\Omega_{n-k-1}(\mathcal{S}_{n,n})$. Moreover, the property of being closed under $\F_q$-multiplication is not preserved under equivalence between two non-linear MRD codes, and then it cannot be used as an argument to state that these codes are not equivalent.\\ 
So, in this section we will investigate their equivalence issue.  Firstly, note that the sets
\begin{equation}\label{gabidulin1}
\mathcal{U}=\left\{ \alpha X^{\sigma^{n-k+1}}+\sum_{i=n-k+2}^{n-1}\beta_i X^{\sigma^i} : \alpha, \beta_i \in \mathbb{F}_{q^n} \right\}=\left\{ f \circ X^{\sigma^{n-k+1}} : f \in \mathcal{G}_{k-1,\sigma} \right\}
\end{equation}
and
\begin{equation}\label{gabidulin2}
\mathcal{V}=\left\{ \alpha X+\sum_{i=n-k+2}^{n-1}\beta_i X^{\sigma^i} : \alpha, \beta_i \in \mathbb{F}_{q^n} \right\}=\left\{ f \circ X^{\sigma^{n-k+2}} : f \in \mathcal{G}_{k-1,\sigma} \right\},
\end{equation}
contained in $\mathcal{C}_{\sigma,T}$, are equivalent to a generalized Gabidulin code $\mathcal{G}_{k-1,\sigma}$ and their intersection
\[
\left\{ \sum_{i=d+1}^{n-1}\beta_iX^{\sigma^i} : \beta_i \in \mathbb{F}_{q^n}\right\}=\left\{ f \circ X^{\sigma^{n-k+2}} : f \in \mathcal{G}_{k-2,\sigma} \right\}
\]
is equivalent to a generalized Gabidulin code $\mathcal{G}_{k-2,\sigma}$.
While the non-linear $(n,n,q;n-k+1)$-MRD code $\mathcal{C}_{n,k,\sigma,I}=\mathcal{C}^{(1)}_{n,k,\sigma,I} \cup \mathcal{C}^{(2)}_{n,k,\sigma,I}$, introduced in Subsection 2.3.,
contains the set 
\[\left\{ \sum_{i=1}^{k-1} \gamma_i X^{\sigma^i} : \gamma_i \in \mathbb{F}_{q^n} \right\}=\left\{ f \circ X^{\sigma} : f \in \mathcal{G}_{k-1,\sigma} \right\},\] which is equivalent to a generalized Gabidulin code $\mathcal{G}_{k-1,\sigma}$ and it is straightforward to see that $\mathcal{C}_{n,k,\sigma,I}$ is contained in a generalized Gabidulin code $\mathcal{G}_{k+1,\sigma}$.\\

Clearly, if $q=2$  or $T=\F^*_q$, $\cC_{\sigma,T}$ is equivalent to the generalized Gabidulin code $\mathcal{G}_{k,\sigma}$. Then, by \cite[Corollary 2.1]{Otal}, it follows 
\begin{thm}\label{thm:equivalence}
Let $I \subseteq \F_{q}$, $1 \in T \subseteq \F_{q}^*$ and $\sigma$ be a generator of $\mathrm{Gal}(\F_{q^n}|\F_q)$.
If $q=2$ or $T=\F_{q}^*$ and $I\in \{\emptyset,\{0\},\F_q^*,\F_{q}\},$  then the codes $\cC_{n,k,\sigma,I}$ and  $\cC_{\sigma,T}$ are both equivalent to $\mathcal{G}_{k,\sigma}$.
\end{thm}

In the remainder of this section, we shall show that apart from these cases, the codes $\cC_{\sigma,T}$ and $\cC_{n,k,\sigma,I}$  cannot be equivalent. In order to do this, we recall the following result that holds for the $\sigma$-linearized polynomials as well.

\begin{thm}\label{Shee}\cite[Theorem 3]{Shee}
A subspace  $\mathcal{W}$ of $\mathcal{G}_{k,\sigma}$, $k \leq n-1$, is equivalent to $\mathcal{G}_{r,\sigma}$ if and only if there exist invertible linearized
polynomials $f, g$ such that 
$$\mathcal{W} = \mathcal{G}^{(f,g)}_{r,\sigma} = \left\{ f \circ \alpha \circ g : \alpha \in \mathcal{G}_{r,\sigma} \right\},$$ where $f_0 = 1$ and $\deg_{\sigma}
(f) + \deg_{\sigma}(g) \leq k - r$.
\end{thm}

By a similar argument used in \cite[Lemma 4]{Shee} for the generalized twisted Gabidulin codes, we prove the following 

\begin{thm}\label{unico}
Let $ I \not\in \{\emptyset,\{0\},\F_q^*,\F_{q}\}$, the non-linear $(n,n,q;n-k+1)$-MRD code $\mathcal{C}_{n,k,\sigma,I}$ contains a unique subspace equivalent to $\mathcal{G}_{k-1,\sigma}$.
\end{thm}

\proof
We have already noted that $\mathcal{C}_{n,k,\sigma,I}$ is contained in $\mathcal{G}_{k+1,\sigma}$. Now, let $\mathcal{W}$ be a subspace of $\mathcal{C}_{n,k,\sigma,I}$ equivalent to $\mathcal{G}_{k-1,\sigma}$. By Theorem \ref{Shee},  there exist invertible $\sigma$-linearized polynomials $f,g$ such that
$$
\mathcal{W}=\mathcal{G}_{k-1,\sigma}^{(f,g)}=\left\{ f \circ \alpha \circ g : \alpha \in \mathcal{G}_{k-1,\sigma} \right\}
$$
 with $f_0 \ne 0$ and $\deg_{\sigma}(f)+\deg_{\sigma}(g) \leq 2$.

\noindent If $\deg_{\sigma}(f)=t$ then $\deg_{\sigma}(g)\leq 2-t$. Let $\alpha=\sum_{i=0}^{k-2}\alpha_i X^{\sigma^i}$ be an element of $\mathcal{G}_{k-1,\sigma}$. Then the coefficient of $X$ in the polynomial $f \circ \alpha \circ g$ is $f_0\alpha_0g_0$, while the coefficient of $X^{\sigma^k}$ is $f_t\alpha_{k-2}^{\sigma^t}g_{2-t}^{\sigma^{k+t-2}}$. Since $\cC^{(1)}_{n,k,\sigma,I}$ and $\cC^{(2)}_{n,k,\sigma,I}$  are disjoint, first suppose that $\mathcal{W} \subset \mathcal{C}^{(1)}_{n,k,\sigma,I}$. Then, $\mathrm{N}_{q^n/q}(f_0\alpha_0 g_0) \in I$ and $f_t\alpha_{k-2}^{\sigma^t}g_{2-t}^{\sigma^{k+t-2}}=0$ for any $\alpha_0, \alpha_{k-2} \in \mathbb{F}_{q^n}$. So, we distinguish three cases:

\begin{itemize}

\item [1.] Suppose $t=2$, then $\deg_{\sigma}(g) = 0$. Since $f_0,f_2 \ne 0$ this is possible if and only if $0 \in I$ and $g_0=0$. This yields $g$ is the zero polynomial, a contradiction.

\item [2.] Suppose $t=1$, then $\deg_{\sigma}(g) \leq 1$. Since $f_0,f_1 \ne 0$ this is possible if and only if $0 \in I$ and $g_1=0$. If $g_0\not= 0$, since the map $x \in \F_{q^n}\mapsto \mathrm{N}_{q^n/q}(f_0xg_0) \in \F_{q}$ is surjective, $I=\F_q$ a contradiction. This yields $g$ is the zero polynomial, leading to a contradiction again.

\item [3.] Suppose $t=0$, then $\deg_{\sigma}(g) \leq 2$. Since $f_0 \ne 0$ this is possible if and only if $0 \in I$ and $g_{2}=0$ and $g_0=0$ as in the previous case. It follows $g=g_1X^{\sigma}$, implying $\mathcal{W}=\left\{ \alpha \circ X^{\sigma}: \alpha \in \mathcal{G}_{k-1,\sigma} \right\}$.

\end{itemize}

\noindent A similar argument can be applied when $\mathcal{W} \subset \mathcal{C}^{(2)}_{n,k,s,I}$. The claim follows.
\endproof

\begin{thm}\label{notequivalent}
Let $I \not \in \{\emptyset,\{0\},\F_q^*,\F_{q}\}$ and $1 \in T \subset \mathbb{F}_q^*$. Then the codes of type $\cC_{n,k,\sigma,I}$ and $\cC_{\sigma,T}$ are neither equivalent nor adjointly equivalent.
\end{thm}

\proof

By contradiction, assume that $\cC_{n,k,\sigma,I}$ and $\cC_{\sigma,T}$ are equivalent or adjointly equivalent. Then, there exist $f,g,h\in \tilde{\mathcal{L}}_{n,q,\sigma}$ with $f,g$ invertible and $\rho \in \Aut(\Fq)$ such that
$$
\cC_{n,k,\sigma,I}=f \circ \cC_{\sigma,T}^{\rho} \circ g+h
\,\,\textnormal{ or }\,\,
\cC_{n,k,\sigma,I}=f \circ (\cC_{\sigma,T}^{t})^{\rho}\circ g+h.
$$
Let  $\mathcal{U,V}$ be the subspaces in $\cC_{\sigma,T}$ equivalent to $\mathcal{G}_{k-1,\sigma}$ as defined in \eqref{gabidulin1} and \eqref{gabidulin2}.  If $\cC_{n,k,\sigma,I}$ and $\cC_{\sigma,T}$ are equivalent, then $f \circ \mathcal{U}^{\rho}\circ g+h$ and $f \circ \mathcal{V}^{\rho} \circ g+h$ are two different subspaces of $\cC_{n,k,\sigma,I}$ equivalent to $\mathcal{G}_{k-1,\sigma}$. By Theorem \ref{unico}, it follows a contradiction.\\
Next, suppose that $\cC_{n,k,\sigma,I}=f \circ (\cC_{\sigma,T}^{t})^{\rho}\circ g +h$. Then $f \circ (\mathcal{U}^{t})^{\rho}\circ  g+h$ and $f \circ (\mathcal{V}^{t})^{\rho} \circ g+h$ are two different subspaces of $\cC_{n,k,\sigma,I}$ equivalent to $\mathcal{G}_{k-1,\sigma}$, see \cite[Lemma 1]{Shee}, a contradiction again.
\endproof

Actually, Theorem \ref{thm:equivalence} and Theorem \ref{notequivalent} can be stated in a slightly more general form, where the automorphisms of $\mathrm{Gal}(\F_{q^n} | \F_q)$ defining the two classes of non-linear MRD codes are not necessarily equal. Indeed, let $\sigma$ and $\tau$ be generators of the group $\mathrm{Gal}(\F_{q^n}|\F_q)$. Recalling that two generalized Gabidulin codes $\mathcal{G}_{k,\sigma}$ and $\mathcal{G}_{k,\tau}$ are equivalent if and only if $\tau   \in \{ \sigma, \sigma^{-1}\}$ (see \cite[Theorem 4.4]{LuTr}), we have the following result
\begin{thm}
Let $\sigma$ and $\tau$ be generators of  the group $\mathrm{Gal}(\F_{q^n}|\F_q)$.
If $q=2$ or $T=\F_q^*$ and $I \in \{\emptyset,\{0 \},\F_q^*,\F_q\}$, then the codes $\cC_{\sigma,T}$ and $\cC_{n,k,\tau,I}$ are  equivalent if and only if $\tau \in \{\sigma, \sigma^{-1}\}$. Otherwise, they are neither equivalent nor adjointly equivalent.
\end{thm}

\begin{remark}
\textnormal{For linear codes, left and right idealisers are invariants under equivalence \cite[Proposition 4.1]{kernelsnuclei} and they are subalgebras of $\tilde{\mathcal{L}}_{n,q,\sigma}$  isomorphic to subfields of $\F_{q^n}$ in the case of square linear MRD codes, see \cite[Theorem 3.1]{LongobardiZanella} and \cite[Theorem 5.4]{kernelsnuclei}. This, as already stated in \cite{kernelsnuclei}, does not hold for non-linear codes. In the following, we shall show an example of two equivalent codes such that their left idealisers have different size.\\
Let $\mathbb{F}_5$ be the field with five elements and let}
\begin{equation}
\mathcal{C}_1=\biggl \{
\begin{pmatrix}
    1 & 2 \\
    3 & 4
\end{pmatrix},
\begin{pmatrix}
  3 & 4\\
  3 & 4
\end{pmatrix} 
\biggr \} \subset \mathbb{F}^{2 \times 2}_5.
\end{equation}
\textnormal{Consider the matrix 
$H=\begin{pmatrix} 
    0 & 0 \\
    2 & 1
\end{pmatrix}$ and the code }
\begin{equation}
    \mathcal{C}_2=\mathcal{C}_1+ H=\biggl \{
\begin{pmatrix}
    1 & 2 \\
    0 & 0
\end{pmatrix},
\begin{pmatrix}
  3 & 1\\
  0 & 0
\end{pmatrix} 
\biggr \}.
\end{equation}
 \textnormal{It is straightforward to see that  $I_L(\mathcal{C}_1)$ is not a matrix field because of the unique element belonging to it is identity matrix, whereas any matrix $I_2$ of the shape}
 \begin{equation*}
\begin{pmatrix}
    1 & a \\
    0  & b
\end{pmatrix}
\end{equation*}
\textnormal{with $a,b \in \mathbb{F}_5$ is an element of $I_L(\mathcal{C}_2)$}.

\end{remark}

\section*{Acknowledgement}
We are grateful to the Italian National Group for Algebraic and Geometric Structures and their Applications (GNSAGA - INdAM) which have supported this research.

\vspace{1cm}
\noindent Dipartimento di Matematica e Applicazioni “Renato Caccioppoli”\\
Università degli Studi di Napoli Federico II,\\
Via Cintia, Monte S. Angelo I-80126 Napoli, Italy. \\
Email addresses: \\
\texttt{\{ndurante, giovannigiuseppe.grimaldi, giovanni.longobardi\}@unina.it}

\end{document}